\documentclass[12pt,a4paper,final]{iopart}

\usepackage{iopams}

 \usepackage{appendix}

\usepackage{graphicx}
\usepackage[breaklinks=true,colorlinks=true,linkcolor=blue,urlcolor=blue,citecolor=blue]{hyperref}
\usepackage{amsthm}
\theoremstyle{plain}
\newtheorem{thm}{Theorem}[section]
 \newtheorem{lem}[thm]{Lemma}

\newtheorem{conjecture}[thm]{Conjecture}

\theoremstyle{definition}
\newtheorem{defn}{Definition}[section]

\usepackage{tikz}
\usetikzlibrary{decorations.pathmorphing}

\makeatletter
\newcommand\footnoteref[1]{\protected@xdef\@thefnmark{\ref{#1}}\@footnotemark}
\makeatother

\begin{document}

\def\mnote#1{{}}

\title[]{Manifold-Topology from K-Causal Order}

\author{Rafael D. Sorkin$^{1, 2}$, Yasaman K. Yazdi$^{3}$, and Nosiphiwo Zwane$^{4}$}
\address{$^1$ Perimeter Institute for Theoretical Physics, 31 Caroline St. N., Waterloo ON, N2L 2Y5,
Canada}
\address{$^2$ Department of Physics, Syracuse University, Syracuse, NY 13244-1130, U.S.A.}
\address{$^3$ Department of Physics, University of Alberta, Edmonton AB, T6G 2E1, Canada}
\address{$^4$ University of Swaziland, Private Bag 4, Kwaluseni, M201, Swaziland}
\ead{rsorkin@pitp.ca, kouchekz@ualberta.ca, ntzwane@uniswa.sz}

\begin{abstract}
To a significant extent, the metrical and topological properties of
spacetime can be described purely order-theoretically. The $K^+$
relation has proven to be useful for this purpose, and one could wonder
whether it could serve as the primary causal order from which everything
else would follow.  In that direction, we prove, by defining a suitable
order-theoretic boundary of $K^+(p)$, that in a $K$-causal spacetime,
the manifold-topology can be recovered from $K^+$.  We also state a
conjecture on how the chronological relation $I^+$ could be defined
directly in terms of $K^+$.
\end{abstract}

\vspace{2pc}


\section{Introduction}
There is much information in the causal structure of spacetime,
including information about the topology, differentiable structure, and
metric (indeed the full metric up to a conformal factor) 
\cite{mal, hawk, Geroch}.  
Accordingly, a field of study that could be called global causal analysis
has grown up to utilize the causal structure directly, beginning with
the well known singularity theorems  \cite{HawkingEllis, Penrose2}.
In \cite{Penrose1} and \cite{Sorkin1} a positive energy theorem was
proven using arguments similar to those that feature in the proofs of
the singularity theorems.
In \cite{Martin:2004xu} the authors show that the causality relation $J^+$
turns a globally hyperbolic spacetime into a bicontinuous poset, which
in turn allows one to define an intrinsic topology called the interval
topology that turns out to be the manifold topology.
Moreover, several approaches to understanding relativity theory, some of
them as old as relativity itself, make primary use of causal structure
(\cite{aa1, aa2}).
Applications of causal structure are far-reaching and extend beyond the
description of spacetime at the classical level.  Causal set theory is
an approach to quantum gravity for which a type of causal structure is
fundamental. The causal set itself is a discrete set of elements
structured as a partial order, and its defining order-relation
corresponds macroscopically to the causal order of spacetime.
(See \cite{CS, CS1, Henson} for more details on causal set theory.)


A natural question is how far can one get with describing manifold
properties using nothing more than order-theoretic concepts. This is of
interest generally, but also in the context of theories such as causal
set theory.  Workers in other fields, like computer science, have been
interested in this question as well \cite{Martin:2004xu}.
But if one wishes to conceive of order-theoretic relationships as the
foundation of spacetime structure,
it seems simplest to 
have in mind
a single order-relation, 
rather than the two or more that one sees in most presentations.\footnote%
{The two most common are $I^+$ and $J^+$.  The ``chronological future''
 of a point-event $p$, $I^+(p)$, is the set of events accessible from
 $p$ by future-directed timelike curves starting from $p$. The ``causal
 future'' $J^+(p)$ is the set of points accessible from $p$ by
 future-directed timelike or null curves starting from $p$.  The pasts
 sets, $I^-(p)$ and $J^-(p)$, are defined analogously, with future
 replaced by past.}
The relation $K^{+}$, defined in \cite{Sorkin1}, was originally
conceived for this purpose.
A closed and transitive generalization of $I^{+}$ and $J^{+}$,
it has properties that make it suitable to serve as the primary causal
order of a spacetime.
Because it is transitive (and acyclic in a $K$-causal spacetime) it
lends itself naturally to order-theoretic reasoning.  And because it is
topologically closed it avoids the problems with supremums that would
arise if one took a relation like $I^{+}$ as basic.

The relation $K^+$ was also designed to be a tool of use when one
attempts to generalize causal analysis to $C^0$ metrics that may fail to
be everywhere smooth and invertible.
In \cite{Sorkin1} Sorkin and Woolgar used $K^+$ to extend certain
results like the compactness of the space of causal curves from $C^2$
Lorentzian metrics to $C^0$ Lorentzian metrics, as needed for the
positive-energy proof in \cite{Penrose1}.
In \cite{Dowker1} Dowker, Garcia, and Surya showed that $K^+$ is robust
against the addition and subtraction of isolated points or metric
degeneracies. This allows such degeneracies to be present in metrics
that contribute to the gravitational path-integral, thus enabling one to
include in the space of histories topology-changing spacetimes
(Lorentzian cobordisms) which contain such degeneracies. This is of
course a very interesting physical consideration. For some recent articles on $K^+$ see
\cite{Saraykar} and \cite{Miller}.


Another line of thought, which comes from causal sets, also points to
the desirability of a sole order-relation (and to $K^+$ as a reasonable
choice thereof).  In a fundamentally discrete context the fine
topological differences that distinguish $I$, $J$, and $K$ from each
other lose their meaning.
On the other hand, 
the discrete-continuum correspondence is normally introduced
(at the kinematic level) 
in terms of so-called  {\it\/sprinklings\/}
which
place
points at random in a Lorentzian manifold  $\mathcal{M}$ via a Poisson
process with a constant density of $\rho=1$ in natural units. 
By definition, the
probability of sprinkling $N$ points into a region with spacetime volume $V$
is $P(N) = \frac{(\rho V )^N}{N!} e^{-\rho V}$.  
This produces a causal
set whose elements are the sprinkled points and whose partial order
relation is most simply taken to be that of the manifold's causal 
order restricted to the sprinkled points. 
Since with a Poisson process, the probability for two sprinkled points
to be  lightlike related vanishes, it makes no difference
which continuum order one uses to induce an order on the sprinkled points. 
However, it is sometimes convenient to consider, instead of a random
sprinkling, something like a ``diamond lattice'' in two-dimensional
Minkowski space, in which case either $J^+$ or $K^+$ would be the
most useful choice. 
In all such cases one loses nothing by thinking of $K^+$ as the basic
continuum-order. 
One sees again how it is more natural to work with only one causal
relation, i.e. one does not distinguish between lightlike and timelike
related pairs of elements, only between causally related and unrelated
pairs.


If $K^{+}$ is to be taken as primitive, then it must be possible in
particular to recover the manifold-topology from it, something which was
not addressed in \cite{Sorkin1}.  We could perhaps obtain the
manifold-topology indirectly by first defining $I^{+}$ in terms of
$K^{+}$, but we will not do this herein (although we will provide a
conjecture suggesting how it might be done).  Instead we proceed
directly from $K^{+}$ to the topology by defining an
{\it\/order-theoretic boundary\/} of $K^{+}(p)$ and demonstrating that
it coincides with the topological boundary. Removing it, we obtain a
family of open sets $A(p,q)$ from which the topology can be reconstructed.

We begin in Section \ref{two} by reviewing the definition and properties
of $K^\pm$.  We then introduce the derived sets $A^\pm(p)$, which we use
throughout this paper.  In Section \ref{three} we show that $A^\pm(p)$
is open and locally equivalent to $I^\pm(p)$.  From this it follows that
the order-interval $A(p,q)$ is locally equal to $I(p,q)$, which
completes the proof that the $K$-causal order yields the manifold
topology.  In Section \ref{four} we state a conjecture for how to obtain
$I^+(p)$ more directly from $K^+(p)$.  The appendix collects the lemmas
from \cite{Sorkin1} that we use in this paper.

Our results do not depend on the spacetime-dimension of the manifold
$\mathcal{M}$ that we work in.

\section{K-Causality}
\label{two}
\subsection{Some Definitions}
\begin{defn}
  $K^+ := \,\prec$ (respectively $K^-$) is the smallest relation that
  contains $I^+$ (respectively $I^-$) and is both transitive\footnote%
{$p \prec q\prec q \prec r\Rightarrow p \prec r$}
and topologically closed \cite{Sorkin1}. 
\end{defn}

\noindent Regarded as a subset of $\mathcal{M}\times\mathcal{M}$, the relation
$K^+$ can be obtained by intersecting all the closed and transitive sets
$R_i$ that include $I^+$ \cite{Sorkin1}:
\begin{equation}
  K^+  = \bigcap_i R_i, \indent \textit{where} \indent  \forall i, I^+\subset R_i\,.
\label{kdef}
\end{equation}
By $K^+(p,\mathcal{M})$ or $K^+(p)$ we denote all the points $q$ such
that $p \prec q$, where $p,q \in \mathcal{M}$.  Let $\mathcal{O}$ be an
open subset of $\mathcal{M}$. For $q\in K^+(p,\mathcal{O})$ we write
$p\prec_\mathcal{O} q$. Figure \ref{example} shows an example of
$K^+(p)$ and how it differs from $I^+(p)$ and $J^+(p)$. In the figure,
$q,r,s \in K^{+}(p)$ while $q\in J^{+}(p)$, $r,s \notin J^{+}(p)$ (and
also $q\in I^{+}(p)$, $r,s \notin I^{+}(p)$).

An open set is $K$\emph{-causal} if and only if $\prec$ induces a
(reflexive) partial ordering on it, in other words iff $\prec$
restricted to the open set is asymmetric.\footnote
{Minguzzi has proven \cite{Minguzzi} that $K$-causality coincides with
  stable causality, and that in consequence, $K^+$ coincides with the
  the Seifert relation $J^+_S$ provided that $K$-causality is in force.}

\begin{figure}
\centering
\begin{tikzpicture}

\draw (0,0) -- (-2,2)

(0.2,-0.07) node {p};

\filldraw [black] (0,0) circle (2pt);

\draw (0,0) -- (4,4.01)

(-0.1,1) node {q};

\filldraw [black] (-0.3,1) circle (2pt);

\draw [decorate,decoration=snake] (-3.1,2) -- (1,2);

\draw (1.05,2.05) -- (-1,4)

(1.05,2) node {{\tiny o}}

(0.2,3) node {r};

\filldraw [black] (0.05,3) circle (2pt);

\draw [decorate,decoration=snake] (-1,4) -- (5.2,3.95);

\draw (-1.05,4.05) -- (-3,6)

(-1,4) node {{\tiny o}};

\draw (-1,4.05) -- (1,6)

(-0.5,5) node {s};

\filldraw [black] (-0.7,5) circle (2pt);

\end{tikzpicture}
\caption{An example of $K^+(p)$. The wiggly lines are removed from the manifold. $q,r,s \in K^{+}(p)$ while $q\in J^{+}(p)$, $r,s \notin J^{+}$.}
\label{example}
\end{figure}
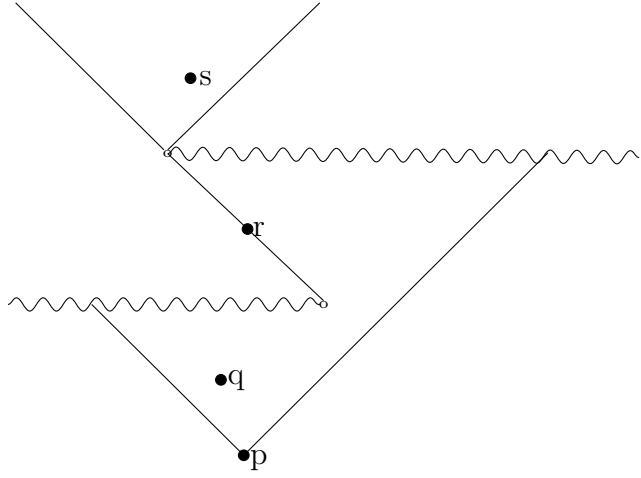

\begin{defn}
The order theoretic or causal boundary of $K^+(p)$, denoted $\partial
K^+(p)$, is the set indicated by
\begin{equation}
  \partial K^+(p)=\{x\in K^+(p)|x=\sup\{y_i\}\, \textrm{for some}\, y_i\notin K^+(p)\}.
  \label{bkdef}
\end{equation}
\end{defn}
\noindent More precisely, a point $x$ of $K^+(p)$ belongs to $\partial K^+(p)$ iff it is the supremum with
respect to $\prec$ of an increasing sequence\footnote%
{Instead of ``increasing sequence'', one could say ``directed set''}
of points $y_i$ all belonging to the complement of $K^+(p)$.
(Intuitively this says that one can approach $x$ arbitrarily closely
from outside of $K^+(p)$, which is a precise order-theoretic counterpart
of the topological definition of boundary.)

\mnote{RDS: the definition given originally lacked the word "increasing", i think we need it but cannot recall why just now}

We next define the open sets $A^+(p)$ which we will use in the next
section to recover the manifold topology.
\begin{defn}
$A^+(p)$ is the set $K^+(p)$ without its causal boundary defined in \eref{bkdef}:
\begin{equation}
    A^+(p) = K^+(p)\backslash \partial K^+(p).
    \label{aplusdef}
\end{equation}
\end{defn}
By definition, $A^+(p)$ is open in the order-theoretic sense. We will
see in the next section that the order-theoretic boundary coincides with
the topological one, therefore making $A^+(p)$ open in the topological
sense as well.

\begin{defn}
The set $A(p,q)$ is 
\begin{equation}
    A(p,q)=A^+(p)\cap A^-(q).
    \label{apq}
\end{equation}

\noindent The set $A(p,q)$ is a kind of ``open order-interval''.  Figure \ref{apq2} shows an example.

\end{defn}

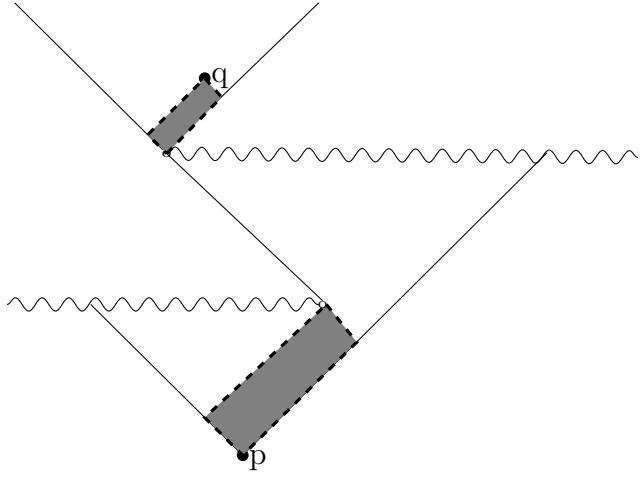
\begin{figure}

\centering

\begin{tikzpicture}
\draw (0,0) -- (-2,2)
(0.2,-0.07) node {p};
\filldraw [black] (0,0) circle (2pt);
\draw (0,0) -- (4,4.01);
\draw [decorate,decoration=snake] (-3.1,2) -- (1,2);
\draw[fill=gray]  (1.1,2) --(1.5,1.5)[dashed,very thick] -- (0,0) -- (-0.5,0.5) -- cycle;
\draw (1.05,2.05) -- (-1,4)
(1.05,2) node {{\tiny o}};
\draw [decorate,decoration=snake] (-1,4) -- (5.2,3.95);
\draw (-1.05,4.05) -- (-3,6)
(-1,4) node {{\tiny o}};
\draw (-1,4.05) -- (1,6)
(-0.3,5) node {q};
\filldraw [black] (-0.5,5) circle (2pt);
\draw[fill=gray]  (-1,4) --(-0.28,4.75)[dashed,very thick] -- (-0.5,5) -- (-1.25,4.25)  -- cycle;
\end{tikzpicture}

\caption{The region in gray is an example of a K-causal open interval $A(p,q)$.}
\label{apq2}
\end{figure}

\section{Manifold Topology from the sets $A(p,q)$}
\label{three}
Let us demonstrate that the sets $A(p,q)$ are locally the same as the
order-intervals $I(p,q)$.  To that end, we first prove that the future-
and past- sets $A^{\pm}(p)$ are topologically open and  locally equivalent to $I^{\pm}(p)$.

To show that $A^+(p)$ is open, it  suffices to show that the
causally defined boundary $\partial K^+(p)$ that we removed from
$K^+(p)$ to yield $A^+(p)$ is also a topological boundary.

{\it  We will assume henceforth that  } $\mathcal{M}$ {\it   is $K$-causal and without boundary. } 

\begin{lem}
$x = \sup {y_i} \Rightarrow x=\lim_{i\rightarrow \infty} \{y_i\}$ in the
  topological sense. Hence if $x$ is (with respect to $K^+$) the
  supremum of an increasing sequence (or directed set) of points $y_i$
  then $x$ is also the limit of $\{y_i\}$ with respect to the manifold
  topology. 
\label{top}
\end{lem}
\begin{proof}
Let the increasing sequence $y_1 \prec y_2 \prec y_3 ...$ have
$x$ as its supremum, and let $U$ be an arbitrarily small open
neighborhood of $x$. From Lemma \ref{lem16} we can assume without loss of
generality that $U$ is K-convex.\footnote%
{A set or subset is said to be K-convex if and only if it contains the
  closed interval $K(p,q)=K^+(p)\cap K^-(q)$ between every pair of its
  elements.}
It suffices to prove that the $y_i$ eventually enter $U$ (after which they
will necessarily remain in $U$ because of the latter's convexity, and
because $y_i \prec x$). 

Now there are two possibilities.  Either the $y_i$ enter $U$ or they
don't. In the former case we are done, so suppose the latter case, and
let $B = \mathrm{Fr}(U)$ be the topological (not causal) boundary of
$U$. The symbol `Fr' stands for ``Frontier'', and we follow the convention of
\cite{dieudonn} in referring to the topological boundary this way.  Since
$U$ is arbitrarily small and $B$ is closed, we can assume without loss of generality that
it is compact.  By Lemma \ref{lem14} there then exist points $z_i$ in $B$
such that ($\forall i$) $y_i \prec z_i \prec x$.  

Recalling that $B$ is compact, and passing if necessary to a
subsequence, we can assume without loss of generality that the $z_i$
converge to some point $z_\infty$ of $B$.  Then since the relation
$\prec$ is closed, and since all of the $z_i \prec x$, we see
immediately that $z_\infty \prec x$.  
On the other hand, if (for fixed $i$) $k > i$,
then we have $y_i \prec y_k \prec z_k$, from which follows $y_i \prec z_k$.  
Therefore we see in the same way from $y_i \prec z_k$  that $y_i
\prec z_\infty$ $\forall i$. 
But by the definition of supremum, $x$ is the least point
of $\mathcal{M}$ for which this holds, hence $x \prec z_\infty$, which
together with $z_\infty \prec x$ implies that they are equal. This is a
contradiction, since $x\in U$ where $U$ is an open set, whereas
$z_\infty\in B$ and $B\cap U=\varnothing$. Therefore $\{y_i\}$ must
enter $U$ and in doing so, they converge (topologically) to $x$, as
desired. 

\end{proof}

\begin{lem}
  Every point $p$ of $\mathcal{M}$ has a neighborhood in which $J^+$ and $K^+$ agree.
\label{lem2}
\end{lem}
\begin{proof}
Given our standing assumption that $\mathcal{M}$ is $K$-causal, Lemma
\ref{lem16} tells us that $p$ has an arbitrarily small $K$-convex neighborhood $U$, 
and then Lemmas \ref{lem12}-\ref{lem13} (with $\mathcal{O}=\mathcal{M}$)
tell us that $K^+$ restricted to $U$
coincides with $K^+$ relativized to $U$ (i.e. computed as if $U$ were the whole
spacetime).  
But we know that in a small enough $U$, the relation $J^+$ is
closed\footnote%
{\label{prop}See Proposition 4.5.1 in  \cite{HawkingEllis} for a proof of this.}, 
transitive, and includes $I^+$.  
It follows  (from the definition of $K^+$) that $J^+(U)$ includes $K^+(U)$, 
and therefore coincides with it (since\footnoteref{prop}
$J^+(U)$ is the closure of $I^+(U)$). 
Notice here, that we can take $J^+$ to be $J^+(U)$.
\end{proof}

\mnote{RDS: Do we need to add here (where we say that J+ is transitive) that U is convex with
 respect to J+, which follows from K-convexity?   Not sure we need, and
 not sure it's worth changing if we do.  Anyhow would be better to prove in general that $J+ \subseteq I+$}

\begin{thm}
The topological boundary of $K^+(p)$ equals its causal
 boundary $\partial K^+(p)$, $\forall p \in \mathcal{M }$.
\label{obtb}
\end{thm}

\begin{proof}
We use the criterion that $x$ is in 
Fr$(K^+(p))$, 
the topological boundary of $K^+(p)$,  
iff
every neighborhood of $x$ contains points both inside and outside of $K^+(p)$.

First let's show that $\partial K^+(p) \subseteq$ Fr$(K^+(p))$. 
Let $x$ be any point of $\partial K^+(p)$.
Since $x$ itself is in $K^+(p)$, it suffices to show that every neighborhood of $x$
contains points in the complement of $K^+(p)$.  This follows directly from
Lemma \ref{top} and the definition of $\partial K^+(p)$.

Conversely let us show that Fr$(K^+(p)) \subseteq \partial K^+(p)$.  
Let $x$ be any point of Fr$(K^+(p))$.

First of all, we easily check that $x\in K^+(p)$, because the
topological boundary of any set lies within its closure, 
and because $K^+(p)$ is closed.

Now let $U$ be a small $K$-convex open neighborhood of $x$ (which exists by Lemma \ref{lem16}).
By Lemmas \ref{lem12}-\ref{lem13}, we can reason as if $U$ were all of $\mathcal{M}$.
And we also know from Lemma \ref{lem2} that if $U$ is sufficiently small
then within it, $K^+$ and $J^+$ coincide.

It's also clear that $I^-(x)$ must be disjoint from $K^+(p)$. Otherwise choose
any $z \in K^+(p) \cap I^-(x)$ and notice that (since $I^+ \subseteq K^+$)
$I^+(z) \subseteq K^+(p)$ would be an open neighborhood of $x$ in $K^+(p)$, hence $x$
would be in $K^+(p)$'s topological interior and not in Fr$(K^+(p))$.

Obviously $I^-(x)$ will contain a timelike increasing sequence of points $y_i$
converging topologically to $x$, 
and this sequence will be disjoint from $K^+(p)$ since $I^-(x)$ is.  
It remains to be proven that $x = \sup \{y_i\}$.

\def\implies{\Rightarrow}

That $x$  bounds the $y_i$ from above is obvious.  
And because the relation `$\prec$' is topologically closed, we have for
any other upper bound $z$ that 
$z\succ y_i \to x \ \implies z \succ x$.
%
Therefore $x$ is a least upper bound, as required. 


\end{proof}

\begin{lem}
  $A^{+}(p)$ is open.  In fact it is the interior of $K^+(p)$.
  \label{aopen}
\end{lem}
\begin{proof}
As defined above in \eref{aplusdef}, $A^+(p)$ is what remains of $K^+(p)$
after we remove its order-theoretic boundary $\partial K^+(p)$. 
But in the theorem just proven we have seen that $\partial K^+(p)$ is
also the topological boundary of $K^+(p)$,
and removing the topological boundary of
any set whatsoever produces its interior, which  by definition is open.
\end{proof}

\def\M{\mathcal{M}}

It follows immediately that the sets $A(p,q)$ are open for all $p$ and $q$. 
We claim furthermore that $A^{+}(p)$ coincides locally with $I^{+}(p)$, 
whence $A(p,q)$ coincides locally with $I(p,q)$.
This follows from Lemmas \ref{lem2} and \ref{aopen}, which inform us that
locally
$A^{+}(p)$ is the interior of $K^{+}(p)$ and  
$K^{+}(p)=J^{+}(p)$. 
To establish our claim, then, simply notice that 
locally 
the interior of $J^{+}(p)$
is $I^{+}(p)$.\footnote%
{This follows from the discussion on pages 33-34 and 103-105 of \cite{HawkingEllis},
 which establishes that the exponential map at $p$ induces a
 diffeomorphism between a neighborhood of $p$ in $\M$ and a neighborhood
 of $0$ in Minkowski-space which preserves the sets
 $I^{+}(p)$ and $J^{+}(p)$.}

We now have all the pieces we need to derive the manifold topology from
the order `$\prec$',
which we do by proving that the open sets $A(p,q)$ 
furnish a basis for the topology of $\M$ 
in the sense that every open subset of $\M$ 
is a union of sets of the form $A(p,q)$.
For this, it suffices in turn that:

\noindent{\bf Criterion:}
for any $x\in\M$ and any open set $U$ containing $x$, we can find $p$
and $q$ such that $x\in A(p,q)\subseteq U$.

\noindent Clearly this criterion is local in the sense that it's enough
for it to hold for $U$ being arbitrarily small.

\begin{thm}
 The sets $A(p,q)$ furnish a basis for the manifold-topology.

\end{thm}

\begin{proof}

Since $K$-causality is in force, strong causality also holds
\cite{Minguzzi}, whereby the manifold-topology is the same as the
Alexandrov topology, for which by definition the sets $I(p,q)$ are a
basis. (See theorem 4.24 in \cite{rogerSIAM}.)
But because locally $I(p,q)=A(p,q)$, as we have just seen,
the sets $A(p,q)$ are also a basis.

To summarize: because the $I(p,q)$ constitute a basis they satisfy the
criterion stated above, and because this criterion is purely local, the
sets $A(p,q)$, which locally coincide with the $I(p,q)$, also satisfy it.

\end{proof}

\noindent{\bf Remark.}
As seen in figure \ref{apq2}, 
there will in general be pairs of points, $p$, $q$, for which the set
$A(p,q)$ 
does not agree with $I(p,q)$.  Such sets
$A(p,q)$ are still included in our basis, but this does no harm, since
by definition, any basis for a topology remains a basis when more sets
are added to it, provided that the additional sets are also open, which
of course the $A(p,q)$ are.

\section{From $K^+(p)$ to $I^+(p)$ and $J^+(p)$}
\label{four}

\def\Reals{\mathbb{R}}

Once we have access to the manifold topology, it is a relatively easy
matter to define continuous curve, and from there to characterize $I^\pm$
and $J^\pm$.  Thus, for example, a curve could be the image of a
continuous function from $[0,1]\subseteq\Reals$ into $\M$, and we might
define a causal (respectively timelike) curve as one which is linearly
ordered by $K^+$ (respectively $A^+$).

\mnote{RDS: alternatively, a curve which is locally within $K^{+}$ (resp $A^+$) of itself.  }

\mnote{RDS: it also seems plausible that $I^+(p)$ is simply the connected component of p in $A^+(p)$}

Nevertheless, it might be nice to characterize $I^+$ and $J^+$ more directly in terms of $K^+$.
We conclude with a conjecture of that nature (concerning $I^+$).

\begin{conjecture} 
  $K^+(p)\backslash S = I^+(p)$, where 
  $S=\{r\in K^+(p) | $ every `full chain' from $p$ to $r$ meets $\partial K^+(p)\}$
\end{conjecture}

\noindent 
Here, by a ``full chain from $p$ to $r$'' we mean a subset $C$ of $\M$ 
containing $p$ and $r$
that is:
 {\it\/linearly ordered\/} by $\prec$ (it is a chain);
 {\it\/order-theoretically closed\/} in the sense that it contains all its suprema and infima;
 and
 {\it\/dense\/} in the sense that it contains between any two of its points a third point 
 [$(\forall x,y\in C)(\exists z\in C)(z\neq x,y  \ \wedge \ x\prec z\prec y)$].

\noindent {\bf Remark.} We didn't need to remove $\partial K^+(p)$
explicitly from $K^+(p)$,  because it is automatically included within $S$.

\bigskip
\noindent\bf Acknowledgements: \rm 
This research was supported in part by Perimeter Institute for Theoretical Physics. Research at Perimeter
Institute is supported by the Government of Canada through the
Department of Innovation, Science and Economic Development Canada and by
the Province of Ontario through the Ministry of Research, Innovation and
Science. YY acknowledges support from the Avadh Bhatia Fellowship at the
University of Alberta. 
This research was supported in part by NSERC through grant RGPIN-418709-2012.

\begin{appendices}

\section{Lemmas from Sorkin-Woolgar (SW) \cite{Sorkin1}}
\label{AppendixA}
In this appendix we collect the lemmas from \cite{Sorkin1} that we have
used in this paper. 

\begin{lem} (Lemma 12 of SW):
 Let $\mathcal{U}$ and $\mathcal{O}$ be open subsets of $\mathcal{M}$ and
 $\mathcal{U}\subseteq \mathcal{O}$. For $p,q\in\mathcal{U}$,
 $p\prec_{\mathcal{U}} q$ implies $p\prec_{\mathcal{O}} q$. 
\label{lem12}
\end{lem}

\begin{lem} (Lemma 13 of SW):
Let $\mathcal{U}$ and $\mathcal{O}$ be open subsets of $\mathcal{M}$ and
$\mathcal{U}\subseteq \mathcal{O}$. For $p,q\in\mathcal{U}$,
$p\prec_{\mathcal{O}} q$ implies $p\prec_{\mathcal{U}} q$, if
$\mathcal{U}$ is causally convex relative to $\prec_\mathcal{O}$. 
\label{lem13}
\end{lem}

\begin{lem} (Lemma 14 of SW):
Let $S$ be a subset of $\mathcal{M}$ with compact boundary $\mathrm{Fr}(S)$, 
and let $x \prec y$ with $x \in S$, $y \notin S$ (or vice versa). 
Then $\exists w \in \mathrm{Fr}(S)$ such that $x \prec w \prec y$. 
\label{lem14}
\end{lem}

\begin{lem} (Lemma 16 of SW):
If $\mathcal{M}$ is K-causal then every element of $\mathcal{M}$ possesses arbitrarily small K-convex open neighborhoods ($\mathcal{M}$ is locally K-convex).
\label{lem16}
\end{lem}

 \end{appendices}

\bibliographystyle{iopart-num}
\bibliographystyle{plain}

\section*{References}

\bibliography{kcausality}
\nocite{*}

\end{document}